\documentclass[10pt,twocolumn]{article}
\usepackage[top=1in, bottom=1in, left=0.85in, right=0.85in]{geometry}
\usepackage{amsmath,amsthm,amssymb}
\usepackage{mathtools}
\usepackage{graphicx}
\usepackage{hyperref}
\usepackage{mathrsfs}
\usepackage{physics}
\usepackage{dsfont}
\usepackage{xcolor}
\usepackage{authblk}
\usepackage{cuted}
\usepackage{balance} 

\usepackage[
  backend=biber,
  style=numeric-comp,
  sorting=none,
  maxbibnames=99,   
  maxcitenames=99   
]{biblatex}
\bibliography{references}

\AtBeginDocument{
  \DeclareSymbolFont{AMSb}{U}{msb}{m}{n}
  \DeclareSymbolFontAlphabet{\mathbb}{AMSb}}

\newcommand{\zz}{\mathbb Z}
\newcommand{\nn}{\mathbb N}
\newcommand{\cc}{\mathbb C}
\newcommand{\cH}{\mathcal H}
\newcommand{\cB}{\mathcal B}

\newcommand{\cO}{\mathcal O}
\newcommand{\cK}{\mathcal K}
\newcommand{\ds}{\mathds{1}}
\newcommand{\id}{\operatorname{id}}
\newcommand{\supp}{\operatorname{supp}}
\newcommand{\lmin}{\lambda^{>0}_{\mathrm{min}}}
\newcommand{\EE}{\mathbb E}
\newcommand{\FF}{\mathbb F}
\newcommand{\WW}{\mathbb W}
\newcommand{\PP}{\mathbb P}

\renewcommand{\epsilon}{\varepsilon}
\renewcommand{\phi}{\varphi}
\renewcommand{\geq}{\geqslant}
\renewcommand{\leq}{\leqslant}

\theoremstyle{plain}
\newtheorem{thm}{Theorem}
\newtheorem{prop}{Proposition}
\newtheorem{lem}{Lemma}

\newtheorem{example}{Example}
\theoremstyle{definition}

\theoremstyle{remark}

\theoremstyle{plain}

\title{Entanglement Sudden Death in Gapped Ground States}
\author[]{Robert Ferydouni\textsuperscript{$*$}}
\author[]{Thomas Jackson\textsuperscript{$\dagger$}}
\author[]{Samuel O.\ Scalet\textsuperscript{$\ddagger$}}
\affil[]{Center for Quantum Mathematics and Physics (QMAP), University of California, Davis, CA 95616, USA}
\affil[]{\textsuperscript{$*$}\texttt{ferydounirobert@gmail.com}\quad
\textsuperscript{$\dagger$}\texttt{thomasjacksonmath@gmail.com}\quad
\textsuperscript{$\ddagger$}\texttt{sscalet@ucdavis.edu}}
\date{}

\begin{document}

\twocolumn[
\begin{@twocolumnfalse}
\maketitle
\begin{abstract}
        Spatial sudden death of entanglement refers to the disappearance of quantum correlations between subsystems after a certain distance called the entanglement length. Specifically, given a density matrix $\rho_{ABC}$ on a tripartite system $ABC$ with $B$ spatially separating $A$ and $C$, the entanglement length is the distance $\ell(A,C)$ after which the reduced density matrix $\rho_{AC}$, given by tracing out the $B$ subsystem, is separable. We show in the case of gapped ground states in any dimension that the entanglement length is finite on finite subsystems. Restricting to one-dimensional systems, we prove that normal matrix product states have a finite entanglement length independent of $A$ and $C$, and that a finite pure state with such uniform entanglement length has a matrix product state representation with bond dimension independent of $A$ and $C$. Consequently, in one dimension, entanglement length distinguishes matrix product states from all other pure states.
\end{abstract}
\vspace{1em}
\end{@twocolumnfalse}
]

\section{Introduction}

Entanglement is the defining nonclassical feature of quantum mechanics and the central resource of quantum information theory.
As a form of correlation between subsystems, it is widely studied in quantum many-body systems, often through measures that capture quantum and classical correlations alike. For ground states of gapped local Hamiltonians, correlations between observables on distant regions decay exponentially~\cite{hastingskoma2006,nachtergaele2006,Kuwahara_2022}, which includes decay of entanglement, but does not distinguish it from the decay of classical correlations. Quantitative entanglement measures for mixed states that can make this distinction vary in motivation, structure, and verifiability~\cite{horodecki2009}. Separability of the state, in contrast, is straightforward: a state is unentangled when it is a convex combination of product states.

In locally interacting quantum systems, correlations between distant observables are normally studied in a way that treats classical and quantum correlations equally. Implicit in this viewpoint is that these correlations are usually nonzero at all distances, as classical correlations typically decay at most exponentially. More recently, a contrasting behavior in quantum entanglement has been demonstrated for noisy channels~\cite{yueberly2004,yueberly2009} and thermal states of many-body systems~\cite{bakshi2024,scalet2026,bakshi2026entanglementquantumspinchains,pilatowskycameo2026onedimensionalquantumgibbsstates}.
In these systems, it has been proven that as a parameter like time, temperature, or distance is varied, the correlations in the system decay asymptotically to zero, whereas entanglement disappears completely at some finite value of the parameter. This phenomenon is called entanglement sudden death.\par
It has long been known that sufficiently small perturbations of the separable, maximally mixed state are separable \cite{zyczkowski1998,gurvitsbarnum2002}. 
Prior work relying on external or thermal noise thus derived sudden death of entanglement by viewing states as small perturbations of the maximally mixed state.\par
In this work, we consider ground states with exponentially decaying (classical and quantum) correlations, such as those of gapped Hamiltonians, and ask whether entanglement between subregions can suddenly die with distance, as previously shown in the one-dimensional thermal setting~\cite{scalet2026}.
These states can no longer be viewed as perturbations of the maximally mixed state, as their marginals are generally not of full rank. While a finite entanglement length has been found analytically \cite{Osterloh_2002, osborne2002}, no general results for local quantum spin systems were known.
Also for specific models, a sudden death of the entanglement negativity, which is weaker than separability, has been observed \cite{Javanmard2018,Parez_2024}.

Our first main result overcomes this difficulty by showing that, for translation-invariant gapped ground states of arbitrary dimension, decay of correlations alone implies spatial sudden death of entanglement. Specifically, for any two subregions $A$ and $C$ of finite size, there exists a finite entanglement length $\ell(A,C)$ such that $\rho_{AC}$ is separable whenever their distance exceeds $\ell(A,C)$. In general the entanglement length may depend on the sizes and shapes of the two regions, a dependence that is unavoidable as the second main result illustrates: we establish a connection between a finite entanglement length and an exact matrix product state representation. For normal matrix product states (those with finite injectivity length), we show there is a finite entanglement length depending only on their transfer operator. We also show that any pure state for which the entanglement length is uniformly upper bounded for all subregions $A, C$ is a matrix product state. Thus, our results show that, in any dimension, entanglement length of gapped ground states on finite subsystems is finite, and in one dimension it distinguishes matrix product states from all other pure states.

\section{Sudden death for gapped ground states}\label{sec:ground_states}
We consider on-site Hilbert spaces $\cH\cong\cc^d$ on the set $\Lambda\subseteq\zz^\nu$. For a finite region $X$ we write $\cH_X=(\cc^d)^{\otimes \abs{X}}$ and $d_X=\dim\cH_X=d^{|X|}$. In this section we consider finite regions $A, C \subset \zz^\nu$ and denote their distance $d(A,C)=\min_{a\in A,c\in C}d(a,c)$, where $d$ is the lattice metric. 
We say two sets $A$, $A'$ are of the same shape if $A'=A+x$ for some $x\in\zz^\nu$. We write $\|\cdot\|$ for the operator norm, $\|\cdot\|_1$ for the trace norm, $\ds$ for the identity matrix and $\id$ for the identity map on operators. 

For a positive semidefinite matrix $\rho$, $\lmin(\rho)$ denotes the smallest nonzero eigenvalue and we let $\mathrm{supp}(\rho) = (\ker \rho)^\perp$. A positive semidefinite operator $X$ on $\cH_A\otimes\cH_C$ is separable if $X=\sum_iX_i^A\otimes X_i^C$ with positive semidefinite $X_i^A$ and $X_i^C$ supported on $A$ and $C$ respectively, and entangled otherwise.

We consider ground states of bounded translation-invariant finite-range Hamiltonians $H_\Lambda=\sum_{X\subset\Lambda} h_X$, which have exponentially decaying correlations \cite{nachtergaelesims2009}; this finite-range assumption enters our arguments only through Lemma~\ref{lem:corr_decay}. We say that a family of Hamiltonians $\{H_\Lambda\}$ has a nondegenerate gapped ground state if, along a sequence of volumes $\Lambda$ exhausting $\zz^\nu$, the ground state $\ket{\Omega_\Lambda}$ of $H_\Lambda$ is unique and separated from the rest of the spectrum by a gap $\gamma>0$ independent of $\Lambda$. All our statements concern the reduced density matrices $\rho^\Lambda_X=\Tr_{\Lambda\setminus X}\ket{\Omega_\Lambda}\bra{\Omega_\Lambda}$ of finitely many spins.
Passing to a subsequence if necessary, we assume the local marginals converge as below, and fix a thermodynamic-limit state
\[
\rho_X=\lim_{\Lambda\to\zz^\nu}\rho^\Lambda_X.
\]
We assume these marginals to be translation-invariant as well.

The following proposition shows that separability arguments about $\rho_{AC}$ can be made in the support subspaces of the marginals $\rho_A$ and $\rho_C$. The main theorem that follows uses the exponential decay of correlations to conclude the bound in the hypothesis of the proposition.

\begin{prop}\label{prop:reduction}
Let $\rho_{AC}$ be a state on $\cH_A\otimes\cH_C$. If
\[
\|\rho_{AC}-\rho_A\otimes\rho_C\|\leq\frac{\lmin(\rho_A)\lmin(\rho_C)}{\sqrt{\operatorname{rank}(\rho_A)\operatorname{rank}(\rho_C)}}\,,
\]
then $\rho_{AC}$ is separable.
\end{prop}

\begin{proof}
Let $P_A$ and $P_C$ be the projections onto $\supp\rho_A$ and $\supp\rho_C$ respectively. Since $\lmin(\rho_A)$ is the smallest nonzero eigenvalue of $\rho_A$, we have $\rho_A-\lmin(\rho_A)P_A\geq0$, and likewise for $C$. Let $\Delta\coloneq \rho_{AC} - \rho_A \otimes \rho_C$. Writing $\rho_A=(\rho_A-\lmin(\rho_A)P_A)+\lmin(\rho_A)P_A$ and similarly for $\rho_C$, we get
\begin{align*}
\rho_{AC}={}&\lmin(\rho_A)\lmin(\rho_C)\\
&\qquad\times\Big(P_A\otimes P_C+\frac{\Delta}{\lmin(\rho_A)\lmin(\rho_C)}\Big)\\
&+(\rho_A-\lmin(\rho_A)P_A)\otimes(\rho_C-\lmin(\rho_C)P_C)\\
&+(\rho_A-\lmin(\rho_A)P_A)\otimes\lmin(\rho_C)P_C\\
&+\lmin(\rho_A)P_A\otimes(\rho_C-\lmin(\rho_C)P_C).
\end{align*}
The last three terms are tensor products of positive semidefinite operators and hence separable. By Lemma~\ref{lem:support} both $\rho_{AC}$ and $\rho_A\otimes\rho_C$, and therefore $\Delta$, are supported in $\supp\rho_A\otimes\supp\rho_C$, on which $P_A\otimes P_C$ acts as the identity. We apply Theorem 1 of \cite{gurvitsbarnum2002} on this support, obtaining separability of $\ds_{\mathrm{supp}(\rho_A) \otimes \mathrm{supp}(\rho_C)}+\Delta'$ when $\|\Delta'\|\le 1/\sqrt{\mathrm{rank}(\rho_A)\mathrm{rank}(\rho_C)}$, where $\lmin(\rho_A)\lmin(\rho_C)\Delta' = \Delta$. This implies $\rho_{AC}$ is separable, as it is a positive linear combination of separable states.
\end{proof}

\begin{thm}\label{thm:ground_state}
Let $\{H_\Lambda\}$ be a family of Hamiltonians with nondegenerate gapped ground states and translation-invariant marginals in the limit. For any finite regions $A,C \subset \zz^\nu$ there is a finite entanglement length $\ell(A,C)$ depending only on the shapes of $A$, $C$ and the Hamiltonian, such that $\rho_{AC}$ is separable if $d(A,C)\ge\ell(A,C)$.
\end{thm}
\begin{proof}
By translation invariance, $\lmin(\rho_A)$, $\lmin(\rho_C)$, $\operatorname{rank}(\rho_A)$ and $\operatorname{rank}(\rho_C)$ are positive constants independent of the positions of $A$ and $C$ in the lattice. Under the gap condition the state has exponential decay of correlations as in Lemma~\ref{lem:corr_decay}:
\[
\|\rho_{AC}-\rho_A\otimes\rho_C\|_1\le K\exp(-\mu d(A,C))\,,
\]
where $\mu>0$ depends only on the gap and on the interaction range and strength, while $K$ additionally depends on $d_A$, $d_C$, $\abs{\partial A}$ and $\abs{\partial C}$. Let
\[
\ell(A,C)=\frac1\mu\max\Big\{1,\,\log\frac{K\sqrt{\operatorname{rank}(\rho_A)\operatorname{rank}(\rho_C)}}{\lmin(\rho_A)\lmin(\rho_C)}\Big\}\,.
\]
Then for $d(A,C) \geq \ell(A,C)$ we have
\begin{align*}
\|\rho_{AC} - \rho_A \otimes \rho_C\| &\leq \|\rho_{AC} - \rho_A \otimes \rho_C\|_1 \\
&\leq \frac{\lmin(\rho_A)\lmin(\rho_C)}{\sqrt{\operatorname{rank}(\rho_A)\operatorname{rank}(\rho_C)}}\,.
\end{align*}
Separability of $\rho_{AC}$ follows from Proposition~\ref{prop:reduction}.
\end{proof}

In the next section we show that for translation-invariant normal matrix product states the entanglement length is independent of $A$ and $C$, i.e., uniform in their size. The converse result, Theorem~\ref{thm:converse}, shows that this dependence cannot be dropped for generic gapped ground states. 

\section{Sudden death for matrix product states}\label{sec:MPS}

We now consider matrix product states in one dimension. Let $V=\cc^D$ be the virtual space of bond dimension $D$. On a finite chain $[1,N]$, we say that a vector $\ket{\Psi}\in(\cc^d)^{\otimes N}$ has a matrix product state representation with bond dimension $D$ if
\begin{equation}\label{eq:mps_rep}
\ket\Psi=\!\!\sum_{i_1,\dots,i_N}\!\!\bra{l}M^{[1]}_{i_1}\cdots M^{[N]}_{i_N}\ket{r}\,\ket{i_1\cdots i_N}
\end{equation} 
for matrices $M^{[k]}_{i}$ of size at most $D\times D$ and boundary vectors $\bra l,\ket r$. 

In the translation-invariant case on the infinite chain we have matrix product state tensors $M_1,\dots,M_d\in\cB(V)$, and we assume the left-canonical gauge $\sum_{i=1}^d M_i^*M_i=\ds_V$. For the normal matrix product states we consider, this can always be achieved by an invertible gauge transformation $M_i\mapsto X M_i X^{-1}$, which leaves the matrix product state invariant; see \cite[Sec.~IV.1]{cirac2021tnreview}. Then the transfer operator
\[
\EE:\cB(V)\to\cB(V),\qquad \EE(Y)=\sum_{i=1}^d M_iYM_i^*,
\]
is a completely positive and trace preserving map (equivalently, a quantum channel). We assume the associated matrix product state is normal, meaning there exists $n$ such that the map $\cB(V) \to \cH_{[1,n]}$ given by
\[X \mapsto \sum_{i_1,\dots ,i_n} \Tr\left(XM_{i_1}\cdots M_{i_n}\right)\ket{i_1 \cdots i_n}\]
is injective. This is equivalent to
\[\cB(V) = \mathrm{span}\{M_{i_1}\cdots M_{i_n} : i_j = 1,\dots ,d\}\,.\]
Proposition~3 of \cite{sanz2010wielandt} shows this is equivalent to primitivity of $\EE$: $1$ is the unique eigenvalue of $\EE$ of modulus $1$ and the corresponding eigenvector $R$ is positive definite and trace-normalized. \par Let $\PP$ be the map $Y \mapsto R\Tr(Y)$, so that $\PP^2 = \PP$. Since $\EE$ is trace preserving and fixes $R$, $\EE \PP = \PP \EE = \PP$. Letting $\WW = \EE - \PP$ we have $\PP\WW=\WW\PP=0$ so that
\begin{equation}\label{eq:spectral}
\EE^n=\PP+\WW^n
\end{equation}
for all $n \in \nn$. 

For an interval $X=[x,x+n-1]\subset\zz$ and a multi-index $\mathbf i=(i_1,\dots,i_n) \in \{1,\dots ,d\}^n$ set $M_{\mathbf i}=M_{i_1}\cdots M_{i_n}$ and define
\[
\FF_X:\cB(V)\to\cB(V)\otimes\cB(\cH_X),
\]
\[
\FF_X(Y)=\sum_{\mathbf i,\mathbf j}M_{\mathbf i}YM_{\mathbf j}^*\otimes\ket{\mathbf i}\bra{\mathbf j},
\]
which is completely positive and trace preserving. The translation-invariant matrix product state generated by $(M_i)$ is the state with reduced density matrices $\rho_X=\Tr_V[\FF_X(R)]$ on intervals $X$. 

In our theorems, we take consecutive intervals $A$, $B$, $C$. Since $M_{\mathbf a\mathbf b\mathbf c} = M_{\mathbf a}M_{\mathbf b}M_{\mathbf c}$ and $\Tr_{\cH_B}\FF_B=\EE^{|B|}$, we obtain
\begin{strip}
\begin{align}
\rho_{AC}&=\Tr_{\cH_B}\Tr_V\big[\FF_{ABC}(R)\big]
=\Tr_V\Tr_{\cH_B}\Big[\big((\FF_A\otimes\id_{BC})\circ(\FF_B\otimes\id_C)\circ\FF_C\big)(R)\Big]\nonumber\\
&=\Tr_V\Big[\big((\FF_A\otimes\id_C)\circ\big(\Tr_{\cH_B}\circ(\FF_B\otimes\id_C)\big)\circ\FF_C\big)(R)\Big]
=\Tr_V\Big[\big((\FF_A\otimes\id_C)\circ(\EE^{|B|}\otimes\id_C)\circ\FF_C\big)(R)\Big].\label{eq:rhoAC_MPS}
\end{align}
\end{strip}

We call a quantum channel $\EE:\cB(\cH)\to\cB(\cH')$ entanglement breaking if $(\id_{\cK}\otimes\EE)(\rho)$ is separable for every auxiliary system $\cK$ and every state $\rho$ on $\cK\otimes\cH$. If $\EE^{|B|}$ is entanglement breaking, the state $$((\EE^{|B|}\otimes\id_C)\circ\FF_C)(R)$$ is separable, and thus $\rho_{AC}$, which is its image of it under a quantum channel, is also separable.

\begin{prop}\label{prop:EB}
If $\EE$ is a primitive channel, there exists $L \geq 1$ such that $\EE^n$ is entanglement breaking for all $n\geq L$.
\end{prop}
\begin{proof}
Let $\ket{\Omega_V}=\sum_{i=1}^D\ket i\otimes\ket i\in V\otimes V$ for an orthonormal basis $\{\ket{i}\}_{i=1}^D$ of the virtual space.
To show that $\EE^{n}$ is entanglement breaking, by \cite[Theorem 4]{horodecki2003}, it suffices to show that the Choi matrix $J_n=(\id_V\otimes\EE^n)(\ket{\Omega_V}\bra{\Omega_V})$ is separable.
By~\eqref{eq:spectral}, and since $(\id_V\otimes\,\PP)(\ket{\Omega_V}\bra{\Omega_V})=\ds_V\otimes R$,
\[
J_n=\ds_V\otimes R+\Delta_n,
\]
\[
\Delta_n=(\id_V\otimes \WW^n)(\ket{\Omega_V}\bra{\Omega_V})\,.
\]
$J_n$ and $R$ are self-adjoint, implying $\Delta_n=J_n-\ds_V\otimes R$ is as well. 

Since any $Y \in \cB(V)$ can be written $Y = \Tr(Y)R + (Y - \Tr(Y)R)$, we have $\cB(V) = \operatorname{Im} \PP \,\oplus\,\ker\Tr$, where the sum is direct because $aR \in \ker\Tr$ implies $a = \Tr(aR) = 0$ so $aR = 0$. On the first summand $\WW$ acts as $0$, since $\PP\WW=\WW\PP=0$, and on the second summand $\WW$ acts as $\EE$. Recalling the primitivity of $\EE$, we therefore have $\sigma(\WW) = \{0\} \cup \left(\sigma(\EE) \setminus \{1\}\right)$, which gives the following for the spectral radius $r(\WW)$ by Gelfand's formula:
\[\|\WW^n\|^{1/n} \to r(\WW) < 1\,.\]
In particular, $\|\WW^n\| \to 0$, so that $\|\Delta_n\| \to 0$. We decompose
\begin{align*}
J_n=&\ds_V\otimes\big(R-\lmin(R)\ds_V\big)\\
&+\lmin(R)\left(\ds_V\otimes\ds_V+\frac{\Delta_n}{\lmin(R)}\right).
\end{align*}
The first term is a tensor product of positive semidefinite operators. Letting $L$ be large enough so that $\|\Delta_n\|/\lmin(R)\leq1/D$ for all $n\geq L$, the second term is separable by Theorem 1 of \cite{gurvitsbarnum2002} applied to the space $V\otimes V$.
\end{proof}

\begin{thm}\label{thm:MPS}
For a normal translation-invariant matrix product state, there exists $\ell\ge1$, depending only on the transfer operator $\EE$, such that $\rho_{AC}$ is separable when $d(A,C)\ge\ell$.
\end{thm}
\begin{proof}
Let $L$ be as in Proposition~\ref{prop:EB} and $\ell = L+1$. If $d(A,C) \geq \ell$ then $\abs{B}\geq L$ so that $\EE^{\abs{B}}$ is entanglement breaking. By normality, $R$ is a full-rank density matrix and $\FF_C$ is a channel, implying $\FF_C(R)$ is a state on $V\otimes\cH_C$. Thus
\[
\sigma=\big((\EE^{|B|}\otimes\id_C)\circ\FF_C\big)(R)
\]
is a separable state on $V\otimes\cH_C$.
By~\eqref{eq:rhoAC_MPS}, $$\rho_{AC}=\big((\Tr_V\circ\FF_A)\otimes\id_C\big)(\sigma)$$ where $\Tr_V\circ\FF_A:\cB(V)\to\cB(\cH_A)$ is a local channel acting only on the virtual bonds, which preserves separability. Thus $\rho_{AC}$ is separable between $A$ and $C$. The non-interval case follows by taking partial traces of subsystems, which do not introduce entanglement.
\end{proof}

As an example, in the appendix we calculate that the AKLT state has entanglement length $2$. We also include an example of a non-normal matrix product state with an entanglement-breaking transfer operator: the GHZ state, which has entanglement length $2$ but for which the transfer operator has degenerate largest eigenvalue.

We call entanglement length scaling which is independent of $A$ and $C$, as in Theorem~\ref{thm:MPS}, uniform entanglement sudden death. The following theorem shows that for pure states it is a characterizing feature of matrix product structure. $L$ is related to the entanglement length by $L = \ell - 1$.

\begin{thm}\label{thm:converse}
Let $\ket\Psi\in(\cc^d)^{\otimes N}$ be a pure state on the chain $[1,N]$ and $L\geq 1$. Suppose that for every decomposition $[1,N]=A\cup B\cup C$ into consecutive nonempty intervals with $|B|=L$, the marginal $\rho_{AC}$ of $\ket\Psi\bra\Psi$ is separable. Then $\ket\Psi$ has a matrix product state representation as in~\eqref{eq:mps_rep} with bond dimension $D\leq d^L$. Furthermore, if a family of pure states $(\Psi_N)_{N\in\nn}$ satisfies the above with $L$ independent of $N$, their bond dimensions are uniformly bounded.
\end{thm}
\begin{proof}
By~\cite{vidal2003}, it suffices to bound the Schmidt rank of $\ket\Psi$ across each cut $[1,k]\,|\,[k+1,N]$ by $d^L$.
If $N-k\leq L$, the Schmidt rank is at most $\dim\cH_{[k+1,N]}\leq d^L$.
Otherwise, let $A=[1,k]$, $B=[k+1,k+L]$ and $C=[k+L+1,N]$, all nonempty.
Since $\ket\Psi$ is pure on $ABC$, the Schmidt decomposition across $B\,|\,AC$ gives $\operatorname{rank}(\rho_{AC})=\operatorname{rank}(\rho_B)\leq d^L$.
As $\rho_{AC}$ is separable by assumption, \cite{horodecki2003rank} shows that $\max\{\operatorname{rank}(\rho_A), \operatorname{rank}(\rho_C)\} \leq \operatorname{rank}(\rho_{AC})$. This implies that the Schmidt rank across $[1,k]\,|\,[k+1,N]$, which equals $\operatorname{rank}(\rho_A)$, is at most $\operatorname{rank}(\rho_{AC})\leq d^L$.
\end{proof}

As a consequence, uniform spatial sudden death on finite chains as in Theorem~\ref{thm:MPS} fails for every family of pure states, such as ground states, whose Schmidt rank across some cut grows with $N$. Therefore, generally the entanglement length in Theorem~\ref{thm:ground_state} must depend on the subregion size. We note the above proof does not require translation invariance, whereas the proof of Theorem~\ref{thm:MPS} does.

\section{Discussion and outlook}\label{sec:outlook}

We discuss the necessity of our hypotheses and some open questions.

\paragraph{The dependence on $\lmin(\rho_A)$.}
The bound on entanglement length from Theorem~\ref{thm:ground_state} depends on the smallest nonzero eigenvalues of the marginals, which can be poorly behaved even for simple models. For the ground state of the transverse-field chain $$H=-\sum_iZ_i+g\sum_iX_iX_{i+1}$$ the single-site marginal is pure at $g=0$ but is expected to have rank two for small $g>0$, and $\lmin(\rho_1)\to0$ as $g\to0$. In this case the bound in Theorem~\ref{thm:ground_state} diverges, although the state at $g=0$ is a product state with zero entanglement.
It is not clear whether the actual entanglement of this family of ground states is bounded as $g\to0$.

\paragraph{Uniformity and the thermodynamic limit.}
Unlike generic gapped ground states, Theorem~\ref{thm:MPS} shows there is no obstruction to defining a finite entanglement length $\ell$ for infinite half-chains in the case of normal matrix product states. Every marginal $\rho_{AC}$ of the left and right half-infinite chains $A$ and $C$ separated by an interval of length $\ell$ is separable, which is the natural notion of separability of the infinite system (see~\cite[Section 4]{scalet2026} for the operator-algebraic formulation).

\section*{Acknowledgements}
This work was supported in part by the National Science Foundation under grants DMS-2510824 (RF \& TJ) and PHY-2337931 (SOS).

\section*{AI Statement}
All content of this paper originated with the authors. Large language models were used as an aid in drafting and editing the exposition. All text has been reviewed and revised by the authors, who take full responsibility for the correctness of the statements and proofs.

\printbibliography

\appendix
\newpage
\section{Appendix}\label{app:rank}

We provide here the lemmas and their proofs.

\begin{lem}\label{lem:support}
Let $\rho_{AC}$ be a state on $\cH_A\otimes\cH_C$ with marginals $\rho_A,\rho_C$.
Then $\supp\rho_{AC}\subseteq\supp\rho_A\otimes\supp\rho_C$.
\end{lem}
\begin{proof}
We show $\ker\rho_A\otimes\cH_C\subseteq\ker\rho_{AC}$.
Let $\ket\psi\in\ker\rho_A$ be nonzero and $P=\ket\psi\bra\psi\otimes\ds_C$.
Then
\begin{align*}
\big\|\rho_{AC}^{1/2}P\big\|_2^2&=\Tr\big(P\rho_{AC}P\big)=\braket{\psi}\,\Tr(P\rho_{AC})\\
&=\braket{\psi}\bra\psi\rho_A\ket\psi=0,
\end{align*}
so $\rho_{AC}^{1/2}P=0$ and hence $\rho_{AC}P=0$.
Since $\ket\psi\otimes\ket\phi=\tfrac{1}{\braket{\psi}}P(\ket\psi\otimes\ket\phi)$ for every $\ket\phi\in\cH_C$, this gives $\ket\psi\otimes\cH_C\subseteq\ker\rho_{AC}$.
Analogously $\cH_A\otimes\ker\rho_C\subseteq\ker\rho_{AC}$.
Taking orthogonal complements and writing $\cH_A=\supp\rho_A\oplus\ker\rho_A$, $\cH_C=\supp\rho_C\oplus\ker\rho_C$,
\begin{align*}
\supp\rho_{AC}&\subseteq(\ker\rho_A\otimes\cH_C)^\perp\cap(\cH_A\otimes\ker\rho_C)^\perp\\
&=(\supp\rho_A\otimes\cH_C)\cap(\cH_A\otimes\supp\rho_C)\\
&=\supp\rho_A\otimes\supp\rho_C,
\end{align*}
where the last equality follows by expanding both spaces in the orthogonal decomposition of $\cH_A\otimes\cH_C$ into the four blocks $(\supp\rho_A\text{ or }\ker\rho_A)\otimes(\supp\rho_C\text{ or }\ker\rho_C)$.
\end{proof}

\begin{lem}\label{lem:corr_decay}
Let $\{H_\Lambda\}$ have a nondegenerate gapped ground state with gap $\gamma$. The bound
\[
\|\rho_{AC}-\rho_A\otimes\rho_C\|_1\leq K_{A,C}\,e^{-\mu d(A,C)}
\]
holds when $d(A,C)\geq\mu^{-1}$, where $\mu$ depends only on $\gamma$, the range $R$ and strength $J$ of the interaction, while $K_{A,C}$ depends on $R, J, d_A, d_C,|\partial A|$ and $|\partial C|$.
\end{lem}
\begin{proof}
Take $\Lambda$ large enough so that it contains $A \cup C$, and write $\langle \cO \rangle = \bra{\Omega_\Lambda} \cO \ket{\Omega_\Lambda}$ for $\cO$ supported on any $X \subset \Lambda$. By duality of the trace norm and operator norm,
\begin{align*}
&\|\rho_{AC}^\Lambda - \rho_A^\Lambda \otimes\rho_C^\Lambda\|_1\\
&\qquad=\sup_{\|\cO\|\leq1}\abs{\Tr\left(\rho_{AC}^\Lambda\cO\right)-\Tr\left(\left(\rho_A^\Lambda\otimes\rho_C^\Lambda\right)\cO\right)},
\end{align*}
with the supremum over operators $\cO$ on $\cH_A \otimes \cH_C$.
Let $\{T_i\}_{i=1}^{d_A^2}$ and $\{S_j\}_{j=1}^{d_C^2}$ be the matrix units of $\cB(\cH_A)$ and $\cB(\cH_C)$. These are orthonormal with respect to the Hilbert--Schmidt inner product and satisfy $\|T_i\|=\|S_j\|=1$. Expanding $\cO=\sum_{i,j}c_{ij}T_i\otimes S_j$,
\begin{align*}
&\|\rho_{AC}^\Lambda - \rho_A^\Lambda \otimes\rho_C^\Lambda\|_1\\
&\qquad\leq\sup_{\|\cO\|\leq1}\sum_{i,j}|c_{ij}|\,\abs{\langle T_i S_j\rangle - \langle T_i \rangle \langle S_j \rangle}.
\end{align*}
Set $\tilde S_j=S_j-\langle S_j\rangle\ds$, so that $\langle T_i\tilde S_j\rangle = \langle T_iS_j\rangle-\langle T_i\rangle\langle S_j\rangle$ and $\|\tilde S_j\|\leq2$. Since $\langle\tilde S_j\rangle=0$ and the ground state is nondegenerate, $\tilde S_j\ket{\Omega_\Lambda}$ lies in the excited subspace of $H_\Lambda$. The exponential clustering theorem~\cite{nachtergaelesims2009} gives that there exists a $\mu > 0$, depending only on $\gamma$, the range $R$ and strength $J$ of the interaction, such that
\begin{align*}
\abs{\langle T_i\tilde S_j\rangle}&\leq\|T_i\|\|\tilde S_j\|\Bigl(1+\frac{1}{\sqrt{\mu d(A,C)}}\\
&\qquad+c_{R,J}\min(|\partial A|,|\partial C|)\Bigr)e^{-\mu d(A,C)}\\
&\leq\frac{K_{A,C}}{(d_Ad_C)^{3/2}}e^{-\mu d(A,C)}\,,
\end{align*}
using $\mu d(A,C)\geq1$, where $c_{R,J}$ depends only on $R$ and $J$, and
\[
K_{A,C} \coloneq 2(d_Ad_C)^{3/2}\Big(2+c_{R,J}\min(|\partial A|,|\partial C|)\Big)\,.
\]
Finally, by Cauchy--Schwarz and orthonormality of the matrix units,
\begin{align*}
\sum_{i,j}|c_{ij}|&\leq d_Ad_C\Big(\sum_{i,j}|c_{ij}|^2\Big)^{1/2}\\
&=d_Ad_C\|\cO\|_2\leq(d_Ad_C)^{3/2}\|\cO\|\,.
\end{align*}
Combined we have the bound
\[\|\rho_{AC}^\Lambda - \rho_A^\Lambda \otimes\rho_C^\Lambda\|_1 \leq K_{A,C} e^{-\mu d(A, C)}\,.\]
Taking the limit $\Lambda \to \zz^\nu$ gives the result. 
\end{proof}

The AKLT state \cite{aklt, affleck1988valence} is the paradigmatic example of a normal matrix product state, and we illustrate a calculation of its entanglement length.\begin{example}
    The AKLT state has entanglement length $2$.
\end{example}
\begin{proof}
The AKLT tensors are \begin{align}
&M^{+}=\begin{bmatrix} 0 & \sqrt{\frac{2}{3}} \\ 0 & 0
\end{bmatrix},\quad M^{0}=\begin{bmatrix} \sqrt{\frac{1}{3}} & 0 \\ 0 & -\sqrt{\frac{1}{3}}
\end{bmatrix},\\& M^{-}=\begin{bmatrix} 0 & 0 \\ -\sqrt{\frac{2}{3}} & 0
\end{bmatrix}.\end{align}
These define the transfer operator: $$\EE(\ds)=\ds,\quad \EE(\sigma^i)=-\frac{1}{3}\sigma^i\,.$$
    From Theorem~\ref{thm:MPS} we evaluate the length $\ell=|B|+1$ at which the iterated transfer operator is entanglement breaking; letting $|B|=1$ we have \begin{equation}
        (\id_V \otimes \EE)(|\Omega_{V}\rangle\langle\Omega_V|)=\frac{2}{3}\ds_4-\frac{1}{3}|\Omega_V\rangle\langle\Omega_V|
    \end{equation}
    which is separable by the PPT criterion \cite{peres1996, horodecki1996}.
\end{proof}
\begin{example} The GHZ state has a non-normal, entanglement-breaking transfer operator and entanglement length $2$.
\end{example}
\begin{proof}
    The GHZ state with on-site Hilbert space $\cc^2$ can be described as a matrix product state with tensors
    \begin{equation}
        M^0=\begin{bmatrix}
            1 & 0 \\ 0 & 0
        \end{bmatrix},\quad M^1=\begin{bmatrix}
            0 & 0 \\ 0 & 1
        \end{bmatrix}
    \end{equation}
    with associated transfer operator \begin{equation}
        \EE(X)=|0\rangle\langle0|X|0\rangle\langle0|+|1\rangle\langle1|X|1\rangle\langle1|
    \end{equation}
    with degenerate eigenvalue $1$.
    Thus \begin{equation}
        (\id_V\otimes\EE)(|\Omega\rangle\langle\Omega|)=\begin{bmatrix}
            1 & 0 & 0 & 0 \\
            0 & 0 & 0 & 0 \\
            0 & 0 & 0 & 0 \\
            0 & 0 & 0 & 1 \\
        \end{bmatrix} \,,
    \end{equation}

    where $|\Omega\rangle$ is the maximally entangled vector. This last matrix is separable as it is a positive sum of product states, so the entanglement length of $\rho$ is bounded with $\ell=2$.
\end{proof}

\end{document}